\documentclass[a4paper,oneside,11pt,reqno]{amsart}
\usepackage{amssymb}
\usepackage{amsthm}
\usepackage{amsmath}
\usepackage{url}
\usepackage{hyperref}
\usepackage{cite}
\usepackage{graphicx}
\usepackage[margin=1in]{geometry}
\usepackage[foot]{amsaddr}
\usepackage{comment}
\usepackage{xcolor}
\usepackage{enumitem}

\newtheorem{theorem}{Theorem}
\newtheorem{corollary}[theorem]{Corollary}
\newtheorem{lemma}[theorem]{Lemma}

\theoremstyle{remark}
\newtheorem*{remark}{Remark}

\newcommand{\ea}[1]{\begin{align}#1\end{align}}
\newcommand{\eat}[2]{\begin{alignat}{#1} #2\end{alignat}}
\newcommand{\nn}{\nonumber \\}
\newcommand{\la}{\langle}
\newcommand{\ra}{\rangle}
\newcommand{\mcl}{\mathcal}
\newcommand{\mbf}{\mathbf}
\newcommand{\mbb}{\mathbb}
\newcommand{\bmt}{\begin{pmatrix}}
\newcommand{\emt}{\end{pmatrix}}
\newcommand{\vpu}[1]{^{\vphantom{#1}}}
\DeclareMathOperator{\Tr}{Tr}
\DeclareMathOperator{\rnk}{Rank}
\DeclareMathOperator{\GL}{GL}
\DeclareMathOperator{\Gal}{Gal}
\DeclareMathOperator{\rel}{Re}
\DeclareMathOperator{\sgn}{Sign}

\begin{document}

\title{Tight frames, Hadamard Matrices and Zauner's Conjecture}



\author{${}^1$Marcus Appleby}
\address{${}^1$Centre for Engineered Quantum Systems, School of Physics, The University of
Sydney, Sydney, NSW 2006, Australia}

\author{${}^2$Ingemar Bengtsson}
\address{${}^2$Stockholms Universitet, AlbaNova, Fysikum, S-106 91 Stockholm, Sweden}

\author{${}^{1,3}$Steven Flammia}
\address{${}^3$Yale Quantum Institute, Yale University, New Haven, CT 06520, USA}

\author{${}^4$Dardo Goyeneche}
\address{${}^4$Departamento de F\'{i}sica, Facultad de Ciencias B\'{a}sicas, Universidad de Antofagasta, Casilla 170, Antofagasta, Chile}


\begin{abstract}
We show that naturally associated to a SIC (symmetric informationally complete positive operator valued measure or SIC-POVM) in dimension $d$ there are a number of higher dimensional structures:  
specifically a projector, a complex Hadamard matrix in dimension $d^2$ and a pair of ETFs (equiangular tight frames) in dimensions $d(d\pm 1)/2$.   We also show that a WH (Weyl Heisenberg covariant) SIC in odd dimension $d$ is naturally associated to a pair of symmetric tight fusion frames in dimension $d$.   We deduce two  relaxations of the WH SIC existence problem.  We also find a reformulation of the  problem in which the number of equations is fewer than the number of variables. Finally, we show that in at least four cases the structures associated to a SIC  lie on continuous manifolds of such structures. In two of these cases the manifolds are non-linear. Restricted defect calculations are consistent with this being true for the structures associated to every known SIC with $d$ between $3$ and $16$, suggesting it may be true for all $d\ge 3$.

\end{abstract}

\maketitle
\allowdisplaybreaks

\section{Introduction}
\label{sec:Intro}
Symmetric informationally complete measurements~\cite{Zauner,Renes} (SIC-POVMs, or SICs as we call them here) are  of much interest to the quantum information and quantum foundations communities~\cite{FHS}.   As a result they have almost become a specialism in itself.  In this paper we restore the balance somewhat by relating them to the larger context of frame theory in general. 

A SIC in dimension $d$ is a collection of $d^2$ unit complex vectors $|\psi_{1}\ra$, \dots, $|\psi_{d^2}\ra$ with the property
\ea{
\big| \la \psi_{j} |\psi_{k}\ra \big|^2 &= \frac{1}{d+1}
\label{eq:SICoverlaps}
}
for all $j\ne k$.  With the single exception of the Hoggar lines~\cite{Hoggar} in dimension 8, every known SIC is covariant with respect to the Weyl-Heisenberg (or WH) group.  WH SICs have been constructed numerically\cite{Zauner,Renes,Scott, Andrew, FHS,Markus2,Markus3} in every dimension up to $181$, and in many other dimensions up to $2208$, while solutions have been constructed~\cite{Zauner,Renes,Scott,Marcus,Markus1,ACFW,Markus2,Markus3,Gene2} in every dimension up to $21$ and in many  
other dimensions up to $1299$ (in both cases, numerical and exact, this listing includes solutions still unpublished).  This lends support to Zauner's conjecture\footnote{This is Zauner's conjecture in its weakest form.  Starting with Zauner himself various stronger versions have been proposed.}:  that SICs exist in every finite dimension.  However, it remains an open question whether that is actually the case.

One of the most striking features singling out SICs from frame theory in general is a remarkable connection with some of the central results and conjectures of algebraic number theory~\cite{AYAZ,AFMY1,AFMY2,Gene2,Gene1}.  The results reported here were originally discovered through exploring that connection.  Specifically, we were first led to Theorem~\ref{cor:condes} as a result of an investigation of the number-theoretical relationships between SICs in different dimensions~\cite{ABDF}.  We were then led to our other main results by working back from that (in other words the order of discovery was the reverse of the order of presentation).   However, having said that let us stress that number theory plays no direct role in this paper.

The common thread which unites the two strands, number-theoretic and geometric, is the role of the squared overlap phases $e^{2i\theta_{jk}}$, where
\ea{
e^{i\theta_{jk}} &= \sqrt{d+1}\la \psi_j |\psi_k\ra
} 
for $j\neq k$.  Since the numbers $e^{i\theta_{jk}}$  are essentially just the elements of the Gram matrix it is obvious that they must play a central role in the description of a SIC.  What is rather less obvious is the fact that their squares appear to be of some independent importance.  Number-theoretically the significance of the squared phases is, firstly the connection with the Stark conjectures discussed in ref.~\cite{Gene2}, and secondly  the fact that they tie together a SIC in dimension $d$ with one in dimension $d(d-2)$, as described in ref.~\cite{ABDF}.    Their geometrical significance comes from the fact that, for any positive integer $t$, the $t^{\rm{th}}$ Hadamard power of the Gram matrix of a projective $t$-design is simply related to a unitary~\cite{Datta}.  Specializing to the case of a  SIC this means that the numbers $e^{2i\theta_{jk}}$ are proportional to the elements of a complex Hadamard matrix.  The purpose of this paper is to tease out some of the  consequences of that fact.  

After some preliminaries in Section~\ref{sec:Prelim} we begin in Section~\ref{sec:associated1} by giving a    simple proof of the fact  that, naturally constructible from the squared overlap phases of a   SIC in dimension $d$ (any SIC, not necessarily a WH SIC), there exist
\begin{enumerate}
\item a rank $d(d+1)/2$ projector in dimension $d^2$,
\item a complex Hadamard matrix in dimension $d^2$,
\item a pair of equiangular tight frames in dimensions $d(d\pm 1)/2$.
\end{enumerate}
The first construction was central to the argument in ref.~\cite{ABDF}, the second is a specialization of the construction in ref.~\cite{Datta}, and the third was originally given by Waldron~\cite{Waldron}.  What is new here is the simplicity of the proofs and the fact that we relate the constructions, both to each other, and to the number theoretic considerations of refs.~\cite{AFMY1,AFMY2,Gene1,Gene2}.

  In Section~\ref{sec:assoc2} we prove a result which is  specific to WH SICs:  namely, that every such SIC is naturally associated to two  tight fusion frames~\cite{Casazza1,Casazza2}.  Associated to these frames is a system of equations which are different from the standard defining equations of WH SIC and which may possibly provide a useful relaxation of the existence problem.   
  
In Section~\ref{sec:naimark} we consider  the Naimark complement of a WH SIC.  This  leads to another relaxation of the existence problem.  It also leads to a reformulation of the existence problem  for which 
\ea{
\lim_{d\to \infty}\left(\frac{\text{number of equations}}{\text{number of real variables}}\right) =1
}
 (unlike the usual formulation, where the limit is infinite).

Finally, in Section~\ref{sec:contfam} we present evidence that the projector, complex Hadamard matrix and equiangular tight frames constructed in  Section~\ref{sec:associated1}, and the tight fusion frames constructed in Section~\ref{sec:assoc2},  lie on continuous manifolds of such objects.  This is in  contrast to the situation with SICs which, aside from the case $d=3$, turn out to be isolated in every case examined~\cite{Scott, Andrew,Markus2,Markus3,Bruzda}.  We give an exact construction of continuous families associated to SIC fiducial orbits $4a$, $6a$, $8b$ (using the Scott-Grassl naming convention~\cite{Scott}) in the main text.  In Appendix~\ref{Ap1} we supplement this result by calculating the restricted defect for the equiangular tight frames in every dimension up to $d=16$ and showing that aside from the case $d=2$, it is always non-zero.  This encourages the conjecture that the structures constructed in this paper lie on continuous manifolds for every value of $d$ greater than $2$.

\section{Preliminaries} 
\label{sec:Prelim}

In finite dimensional Hilbert spaces a frame can be defined as a finite set of vectors 
that span the space. A frame is said to be unit-norm if it consists entirely of unit 
vectors. 
A unit-norm frame is said to be tight if and only if it forms a resolution of the identity,   
\ea{
\sum_{j=1}^n |\psi_j\ra \la\psi_j| &= \frac{n}{d} \boldsymbol{1}_d
 \label{tf} 
}
where $n$ is the number of vectors and $d$ is the dimension. Up to normalization a tight frame is the same thing as a rank one POVM. Finally, an equiangular tight frame (ETF) is defined to be a  unit-norm tight frame for which  $|\la \psi_j |\psi_k\ra| = c$ for some fixed constant $c$ and all $j\neq k$.  Squaring both sides of Eq.~\eqref{tf} and taking the trace we see that
\ea{
c&= \sqrt{\frac{n-d}{d(n-1)}}.
\label{values}
}
A simple example of a real ETF is the frame defined by the vertices of a regular tetrahedron.  A less easily visualized example is the six vectors obtained by taking one of each pair of diametrically opposite vertices of a regular icosahedron.  

The fact that an ETF is a spanning set means it must contain at least  $d$ vectors.    It is also easily seen that the number of vectors cannot exceed $d^2$.  Indeed, suppose it did.  Then the projectors $|\psi_j\rangle \langle \psi_j|$ would have to be linearly 
dependent in the sense that $\sum_{j} \lambda_j |\psi_j\ra \la \psi_j| =0$ for some 
$\lambda_j$ not all zero.  Multiplying both sides by $|\psi_k\ra \la \psi_k|$  and  taking the trace gives
\ea{
(1-c^2) \lambda_k + c^2  \sum_{j=1}^n\lambda_j &=0 
}
where $c$ is the number specified by Eq.~\eqref{values}.  Since this has to be true for all $k$, and since $c^2 \neq 1$, it would follow that the $\lambda_j$ are all zero, contrary to assumption.
A SIC is  defined to be an ETF for which  
 the number of vectors equals the maximum value $d^2$.

For every unit-norm frame  $|\psi_1\ra, \dots , |\psi_n\ra$ in $d$ dimensions, and every positive integer $t$, there holds the Welch bound~\cite{Welch}
\ea{
\binom{d+t-1}{t} \sum_{j,k=1}^n \big| \la \psi_j |\psi_k\ra\big|^{2t} \ge n^2
}
This inequality will play a central role in the sequel.  It can be shown~\cite{KR} that a unit-norm  frame saturates the Welch bound if and only if it is a complex projective $t$-design (from which it follows that if the bound is saturated for one value of $t$, it is saturated for all smaller values).  In particular, a tight frame is a complex projective $1$-design, and a SIC is a complex projective $2$-design. 

As mentioned in Section~\ref{sec:Intro} all but one of the known SICs is covariant with respect to the WH group.  To describe this group we use the conventions of ref.~\cite{Marcus}.  Let $|0\ra, \dots, |d-1\ra$ be an orthonormal basis for $d$ dimensional Hilbert space $\mcl{H}_d$ and for each $\mbf{p}=(p_1,p_2) \in \mbb{Z}^2$ define the displacement operator $D_{\mbf{p}}$ by
\ea{
D_{\mbf{p}} |j\ra &= \tau^{p_2(p_1+2j)}|j+p_1\ra, & \tau&=-e^{\frac{\pi i}{d}}
\label{eq:dopprops}
}
where addition of ket labels is mod $d$. 
 Note that in  even dimensions $\tau^d=-1$, which means   $D_{\mbf{p}}= \pm D_{\mbf{p}'}$ when $\mbf{p}=\mbf{p}'$ mod $d$, but that equality is not guaranteed unless $\mbf{p}=\mbf{p}'$ mod $2d$.  In all dimensions, even or odd, $\tau$ satisfies the identity
 \ea{
 \sum_{\mbf{p}} \tau^{\la \mbf{p},\mbf{q}\ra} &= d^2\delta_{\mbf{q},\boldsymbol{0}}
 }
 where $\la \cdot, \cdot \ra$ is the symplectic form defined by $\la \mbf{p}, \mbf{q}\ra = p_2q_1-p_1q_2$.  
 It is easily seen that the displacement operators satisfy
\ea{
D^{\dagger}_{\mbf{p}}&=D^{\vphantom{\dagger}}_{-\mbf{p}} & D_{\mbf{p}} D_{\mbf{q}} &= \tau^{\la \mbf{p},\mbf{q}\ra } D_{\mbf{p}+\mbf{q}}
\label{eq:whdopprops}
} 
These facts, together with the fact that $\Tr(D_{\mbf{p}}) = d\delta_{\mbf{p},\boldsymbol{0}}$, mean that the set $\{d^{-\frac{1}{2}} D_{\mbf{p}} \colon p_1, p_2 = 0, \dots d-1\}$ is an orthonormal basis for operator space relative to the Hilbert-Schmidt inner product \cite{Schwinger}.  Consequently an arbitrary operator $A$ can be uniquely expanded
\ea{
A & = \sum_{p_1,p_2=0}^{d-1} a_{\mbf{p}} D_{\mbf{p}} ,  & a_{\mbf{p}} &= \frac{1}{d} \Tr(D_{-\mbf{p}} A).
\label{eq:aexpn}
}
We now define a WH SIC  to be  the set  
$\{D_{\mbf{p}}|\psi \ra \colon p_1,p_2=0, \dots, d-1\}$ where $|\psi\ra$, the fiducial vector, is any vector satisfying the $d^2$ equations
\ea{
\big| \la \psi |D_{\mbf{p}}|\psi\ra \big|^2&=
 \begin{cases}
 \frac{1}{d+1} \qquad & \mbf{p} \neq \boldsymbol{0} \mod d
 \\
 1  \qquad & \mbf{p} = \boldsymbol{0} \mod d
\end{cases}
\label{eq:sicconds}
}

To a considerable extent we can replace the study of frames with the study of their Gram matrices.  Recall that the Gram matrix of a set of $n$ vectors $|\psi_k\ra$ is defined to be the $n\times n$ matrix $G$ whose matrix elements are the inner products:
\ea{
G_{jk} &= \la \psi_j |\psi_k\ra.
}
We now  turn this round, and characterize the systems of vectors  under considerations in terms of their Gram matrices.  In the first place, the necessary and sufficient condition for an $n\times n$ Hermitian matrix $G$, to be the Gram matrix of a $d$-dimensional frame  is that it be positive semi-definite and rank $d$.  The necessary and sufficient condition for the frame to be tight is that $(d/n) G$ be a projector.  The frame is, in addition, equiangular if and only if the diagonal elements of $G$ are all $1$, and the off-diagonal elements all have the same absolute value.  Finally, two different frames have the same Gram matrix if and only if they are unitarily equivalent.  The use of these facts is that they mean that, instead of having to look at an entire orbit under the unitary group, one can focus on the single matrix\footnote{Although the Gram matrix is invariant under the replacement $|\psi_r\ra \to U|\psi_r\ra$, where $U$ is any unitary, it is not invariant under arbitrary rephasings $|\psi_r \ra \to e^{i\theta_r} |\psi_r\ra$. As discussed in ref.~\cite{AFF}, the  triple products
 \ea{
 T_{rst} &= G_{rs} G_{st} G_{tr},
 }
provide a characterization which is invariant with respect to both unitary transformations and rephasing.
 } $G$.
 
Let $G$ be the Gram matrix  of an ETF, and define
\ea{
\tilde{G} &= \frac{n}{n-d} \boldsymbol{1}_n - \frac{d}{n-d} G
}
The fact that $(d/n) G$ is a rank $d$ projector means that  $((n-d)/n) \tilde{G}$ is a rank $(n-d)$ projector.  So $\tilde{G}$ is the Gram matrix for an $(n-d)$ dimensional frame, called the Naimark complement~\cite{NaimComp,HL} of $G$ (we say ``the complement'' to save words, although there is, of course, a whole unitary  orbit of complements).  

Finally, we will need the concept of a tight fusion frame~\cite{Casazza1,Casazza2}. In the language of quantum information theory, this is a POVM of arbitrary rank. Instead of thinking of a tight frame as a set of vectors, one can  identify it with the rank 1 projectors appearing in the sum on the left hand side of Eq.~\eqref{tf}.  A tight fusion frame is the generalization of this to a set of projectors   $\Pi_1$,\dots, $\Pi_n$, not necessarily rank 1, which sum to a multiple of the identity.  We will say the frame is a symmetric tight fusion frame (STFF) if the projectors all have the same rank and if $\Tr(\Pi_j \Pi_k) =c$ for some fixed constant $c$ and all $j\neq k$ (in the literature~\cite{FJMW,KingSTFF,KPCL} these structures are referred to as equichordal tight fusion frames,  or  equidistant tight fusion frames; we use a different term because we wish to stress the analogy with a SIC).  A symmetric tight fusion frame is thus a generalization of an ETF.  Unlike an ETF, there is no lower bound on the number of elements of a STFF (the single element set consisting just of the identity is such a frame).  However, the same argument which shows that the number of vectors in an ETF cannot exceed $d^2$ also applies to a STFF.  Up to  normalization a STFF for which $n$ achieves its maximum value of $d^2$ is a particular instance of what in ref.~\cite{Graydon} is called a SIM (or symmetric informationally complete measurement of arbitrary rank).  It is also a conical 2-design~\cite{Graydon}.  It should, however, be  observed that whereas every SIC is an ETF, not every SIM is an STFF.  In the sequel we will have occasion to consider WH covariant STFFs.  These are constructed in the same way as WH covariant ETFs, by choosing a fiducial projector $\Pi$ and then transforming it to obtain the set $\{\Pi_{\mbf{p}} \colon p_1,p_2=0,\dots, d-1\}$.

\section{The Hadamard and Equiangular Tight Frames Associated to a SIC}
\label{sec:associated1}
In this section we show that associated to a SIC in dimension $d$ there always exists
\begin{enumerate}
\item a rank $d(d+1)/2$ projector $Q$ in dimension $d^2$,
\item a complex Hermitian Hadamard matrix $H$ in dimension $d^2$,
\item two ETFs $E$ and $\tilde{E}$ in dimensions $d(d+1)/2$ and $d(d-1)/2$.
\end{enumerate}
  As we noted in the Introduction these objects were originally discovered in the course of the number-theoretic investigations described in 
ref.~\cite{ABDF}, but in this section we do not even use the Weyl--Heisenberg group. We only use 
the definition of a SIC as given in Eq. (\ref{eq:SICoverlaps}). 
\begin{theorem}\label{tm:main}
Let $G$ be the Gram matrix of a SIC in dimension $d$, and let $G^{(2)} = G\circ G$ be its Hadamard square.  Then 
\ea{
Q &= \frac{d+1}{2d} G^{(2)}
\label{eq:Qdef}
}
is a rank $d(d+1)/2$ projector in dimension $d^2$.
\end{theorem}
\begin{remark}
The symbol $\circ$ stands for the 
entrywise Hadamard (or Schur) product. $G^{(2)}$ is thus the matrix obtained by squaring all the matrix elements of $G$.  The theorem is a specialization of results   in refs.~\cite{Belovs,Datta}.  However, we feel there is some merit to giving a free-standing proof.  We particularly wish to draw attention to the fact that the theorem has a subtle relationship with the Welch bound in that, although it is  not immediately and obviously implied, it falls out of a suitable proof of the latter  as a kind of side benefit. 
\end{remark}
\begin{proof}
Let $|\psi_1\ra, \dots , |\psi_n\ra$ be any unit-norm frame in $\mbb{C}^d$ (not necessarily a SIC),  let $G$ be its Gram matrix, and let $G^{(2)}$ be the Hadamard square of $G$. 
$G^{(2)}$ is a Hermitian matrix of trace $n$, and has real eigenvalues 
that we denote by $\lambda_j$. We see that 
\begin{equation} \sum_{j=1}^n\lambda_j^2 = \Tr \left(G^{(2)}G^{(2)}\right) = 
\sum_{j,k}|\langle \psi_j|\psi_k\rangle |^4 \ . \end{equation} 
\noindent Now we recall the inequality 
\begin{equation} \lambda_1^2 + \dots + \lambda_r^2 
\geq \frac{1}{r}(\lambda_1 + \dots + \lambda_r)^2 \ . \label{olik1} \end{equation}

\noindent It is saturated if and only if all the $\lambda_j$ are equal to each 
other. Only the non-zero eigenvalues count, so we set $r = \rnk \left(G^{(2)}\right)$. 
Then we deduce that 

\begin{equation} \sum_{j,k}|\langle \psi_j|\psi_k\rangle |^{4} \geq 
\frac{\left(\Tr\left( G^{(2)}\right)\right)^2}{\rnk(G^{(2)})} = 
\frac{n^2}{\rnk\left(G^{(2)}\right)}\ . 
\label{Welchprel} \end{equation}  

\noindent It remains to find an upper bound on the rank of the matrix $G^{(2)}$. 
This is done by introducing the generator matrix $X$ whose $n$ columns are the 
vectors $|\psi_j\rangle$ we started out with. Then we reinterpret $X$ as a 
matrix consisting of $d$ row vectors $\langle w_j|$ in $\mbb{C}^n$, with 
components $w^{*}_{ja} = \la w_j | a\ra = \la j |\psi_a\ra$. We will soon be considering Hadamard products of 
such vectors, with components $(w_j \circ w_k)_a = w_{ja}w_{ka}$. The 
Gram matrix can be written as 

\begin{equation} G = X^\dagger X = \sum_{j=1}^d|w_j\rangle \langle w_j| \ . 
\end{equation}

\noindent Its Hadamard square has the matrix elements 

\begin{equation} (G^{(2)})_{ab} = \left(w^{\vphantom{*}}_{1a}w^*_{1b} + w^{\vphantom{*}}_{2a}w^*_{2b} + 
\dots + w^{\vphantom{*}}_{da}w^*_{db}\right)^2 \ . \end{equation} 

\noindent Expanding this out we see that  $G^{(2)}$  is the sum of $d(d+1)/2$ rank 1 operators, from which it follows that its rank cannot exceed $d(d+1)/2$.   
Using this fact in (\ref{Welchprel}) gives the Welch bound when 
$t = 2$. (The proof of the Welch bound for arbitrary $t$ is similar \cite{Belovs}.) 
Thus the $t = 2$ Welch bound is saturated if and only if $G^{(2)}$ is a matrix 
of rank $d(d+1)/2$ all of whose non-zero eigenvalues are equal. When 
$\Tr\left(G^{(2)}\right) = d^2$ they are in fact equal to $2d/(d+1)$. Theorem 1 follows, because 
we also know that the Welch bound is saturated by the $d^2$ vectors in a SIC.  
\end{proof}
This theorem is the key result on which the rest of this paper depends.  Specifically, it means
\begin{enumerate}
\item \label{it:had} Starting from a SIC in dimension $d$ one can construct an order $d^2$ complex Hadamard matrix  (Corollary~\ref{cor:had} in this section).
\item Starting from a SIC in dimension $d$ one can construct a pair of ETF's in dimensions $d(d\pm 1)/2$ (Corollary~\ref{cor:ETF} in this section).
\item Starting from a WH covariant SIC in odd dimension $d$ one can construct a pair of rank $(d\pm 1)/2$  STFF's in dimension $d$ (Theorem~\ref{cor:condes} in Section~\ref{sec:assoc2}).
\end{enumerate}
In connection with~\eqref{it:had}, recall that an order $n$ complex Hadamard matrix is an $n\times n$ unitary matrix all of whose elements have the same absolute value (necessarily equal to $1/\sqrt{n}$).

\begin{corollary}\label{cor:had}
Let $G$ be the Gram matrix of a SIC in dimension $d$, and let $Q$ be the projector defined in Eq.~\eqref{eq:Qdef}.    Then
\ea{
H = 2 Q - \boldsymbol{1}_{d^2}.
}
is a complex Hermitian Hadamard matrix with positive diagonal and $\Tr(H) = d$.  
\end{corollary}
\begin{proof} The fact that $Q$ is a projector implies that $H$ is a Hermitian unitary.  The definition of $Q$ in terms of the Gram matrix of a SIC means that the matrix elements of $H$ all have the same absolute value, and that the ones on the diagonal are all positive equal to $1/d$.  The fact that $\Tr(H)=d$ is immediate. 
\end{proof} 

In spite of much interest the classification problem for complex Hadamard matrices  remains unsolved~\cite{TZ,Feri2}.  It is, however, known that 
complex Hadamard matrices exist in every dimension (the Fourier matrix being an obvious example). Moreover continuous families of Hadamard matrices are known to exist in many dimensions. A continuous family is said to be affine if the variable phases occurring in the matrix elements are either free, or affine functions of free phases~\cite{TZ}. Affine families exist for instance in all composite dimensions. What is more relevant for our purposes is that a small number of continuous \emph{non}-affine families are known. See ref.~\cite{Karlsson} for an example that will figure briefly below.  In the following we will add a few more examples to the list.
\begin{corollary}\label{cor:ETF}
 Let $G$ be the Gram matrix of a  SIC in dimension $d$, and let $Q$ be the projector defined in Eq.~\eqref{eq:Qdef}.    Then
 \ea{
E &= \frac{2d}{d+1}Q,
\\
\tilde{E} &= \frac{2d}{d-1}\left(\boldsymbol{1}_{d^2}-Q\right)
}
are the Gram matrices of a pair of Naimark complementary ETFs in dimension $d(d+1)/2$, $d(d-1)/2$ respectively.
\end{corollary}
\begin{remark}
The existence of these  ETFs was first proved by Waldron \cite{Waldron}.
\end{remark}
\begin{proof} As discussed in Section~\ref{sec:Prelim}, the fact that  $Q$ is a rank $d(d+1)/2$ projector means that $E$ is the Gram matrix of a tight frame in 
$\mbb{C}^{\frac{d(d+1)}{2}}$.  The fact that the diagonal elements are all equal to 1, and the off-diagonal elements all have the same modulus, means 
the frame is equiangular.  The fact that $\tilde{E}$ is  the Naimark complement of $E$ means it is the 
Gram matrix of an ETF in $\mbb{C}^{\frac{d(d-1)}{2}}$.  
\end{proof}

It is known that complex Hermitian Hadamard matrices with constant diagonal can be 
be used to construct ETFs~\cite{Holmes, Feri,Turek}. One noteworthy construction uses complex Hadamard matrices containing roots of unity only~\cite{Elwood}. ETF's of the particular kind we just encountered can also be constructed using combinatorial ideas and are then called Steiner ETF's~\cite{Mixon}. 
\subsection{Further remarks:  the special case $d=3$.}
In Corollary~\ref{cor:had} we saw that a SIC in dimension $d$ naturally gives rise to a $d^2 \times d^2$ complex Hermitian Hadamard matrix with positive diagonal.   In the special case $d=3$  the connection is two-way: 
\begin{theorem}\label{tm:turek}
Let $H$ be an arbitrary $9\times 9$ complex Hermitian Hadamard matrix with positive diagonal, and let
\ea{
G &= \frac{3}{2}(\boldsymbol{1}_9-H).
}
Then $G$ is the Gram matrix of a SIC in dimension $3$. 
\end{theorem}
\begin{proof}
See ref.~\cite{Turek}.
\end{proof}
Combined with the results proved above this implies
\begin{corollary}
Let $H$ be an arbitrary $9\times 9$ complex Hermitian Hadamard matrix with positive diagonal, and let $H^{(2)}$ be its Hadamard square.  Then $3H^{(2)}$ is another complex Hermitian Hadamard matrix with positive diagonal.
\end{corollary}
\begin{proof}
It follows from Theorem~\ref{tm:turek} that $G=3(\boldsymbol{1}_9-H)/2$ is the Gram matrix of a SIC in dimension 3.  In view of  Theorem~\ref{tm:main} and Corollary~\ref{cor:had} this means
\ea{
\tilde{H}&= \frac{4}{3} G^{(2)} -\boldsymbol{1}_{9} = 3H^{(2)}
}
is a complex Hermitian Hadamard matrix with positive diagonal.
\end{proof}
\section{Symmetric Tight Fusion Frames Associated to a SIC}
\label{sec:assoc2}
In the previous section nothing was assumed about the SIC, except that it is a SIC.  In this section we specialize to the case of a WH SIC in odd dimension $d$.  We show that the results proved in Section~\ref{sec:associated1} mean that the existence of such a SIC  implies the existence of two symmetric tight fusion frames (STFFs) in the same dimension $d$.

 Let $|\psi\ra$ be a fiducial vector for the SIC, and let $\Pi = |\psi\ra \la \psi|$ be the corresponding rank 1 projector.  
 Then, using the expansion of Eq.~\eqref{eq:aexpn}, we find
\ea{
\Pi & = \frac{1}{d} \boldsymbol{1}_d +\frac{1}{d\sqrt{d+1}} \sum_{\mbf{p}\neq 0} e^{i\theta_{\mbf{p}}} D_{\mbf{p}}
}
where
\ea{
e^{i\theta_{\mbf{p}}} =\begin{cases}
1 \qquad & \mbf{p} = \boldsymbol{0} \mod d
\\
 \sqrt{d+1} \Tr(D^{\dagger}_{\mbf{p}}\Pi) \qquad & \text{otherwise}
 \end{cases}
\label{eq:phsedef}
}
We will show that, if we  replace  the phases $e^{i\theta_{\mbf{p}}}$ in this formula with their squares, and if we make a few other adjustments, then we obtain the fiducial projectors for  two STFFs.  The proof relies on the following Lemma.
\begin{lemma} \label{lm:stffcond}
In odd dimension $d$ let 
\ea{
A &= \frac{1}{d} \sum_{p_1,p_2=0}^{d-1} e^{i\phi_{\mbf{p}}} D_{\mbf{p}}
}
where the phases $e^{i\phi_{\mbf{p}}}$ are such that $e^{i\phi_0} = 1$, $e^{i\phi_{-\mbf{p}}} = e^{-i\phi_{\mbf{p}}}$ for all $\mbf{p}$,  $e^{i\phi_{\mbf{p}}} = e^{i\phi_{\mbf{q}}}$ for all $\mbf{p} = \mbf{q}$ (mod $d$),  and
\ea{
\sum_{\mbf{u}} \tau^{\la \mbf{u}, \mbf{p}\ra} e^{i(\phi_{\mbf{u}}+\phi_{\mbf{p}-\mbf{u}})} &= d^2\delta^{(d)}_{\mbf{p},\boldsymbol{0}},
\label{eq:Acond}
}
the quantity $\delta^{(d)}_{\mbf{p},\boldsymbol{0}}$ being unity if $\mbf{p}=\boldsymbol{0}$ (mod $d$) and zero otherwise.  Then
\ea{
\Pi^{\pm} & = \frac{1}{2} \left( \boldsymbol{1}_d \pm A \right)
\label{eq:pipmdf}
}
are the fiducial projectors for a pair of WH covariant STFFs of rank $(d\pm 1)/2$.
\end{lemma}
\begin{remark}
The requirement that $d$ is odd is essential, since otherwise $A$, and consequently $\Pi^{\pm}$, would not be Hermitian.
\end{remark}
\begin{proof}
$A$ is an Hermitian operator such that $\Tr(A) = 1$ and 
\ea{
A^2 
&= \frac{1}{d^2} \sum_{\mbf{u},\mbf{p}} \tau^{\la \mbf{u},\mbf{p}\ra}e^{i(\phi_{\mbf{u}} + \phi_{\mbf{p}-\mbf{u}})}  D_{\mbf{p}} = \boldsymbol{1}_d.
}
It follows that the operators $\Pi^{\pm}$ defined by Eq.~\eqref{eq:pipmdf} are rank $(d\pm 1)/2$ projectors.  We also have
\ea{
A^{\vphantom{\dagger}}_{\mbf{p}} &= D^{\vphantom{\dagger}}_{\mbf{p}} A D^{\dagger}_{\mbf{p}} 
=\frac{1}{d} \sum_{\mbf{u}} e^{i\phi_{\mbf{u}}}\tau^{2\la \mbf{p},\mbf{u}\ra} D^{\vphantom{\dagger}}_{\mbf{u}}.
}
It follows that
\ea{
\sum_{\mbf{p}} A_{\mbf{p}} &= d\boldsymbol{1}_d,
}
and, consequently,
\ea{
\sum_{\mbf{p}} \Pi^{\pm}_{\mbf{p}} &= \frac{1}{2}d(d\pm1 ) \boldsymbol{1}_d.
}
So the sets $\{\Pi^{\pm}_{\mbf{p}}\colon p_1,p_2 = 0, \dots, d-1\}$ are tight fusion frames.   Finally
\ea{
\Tr(A_{\mbf{p}}) &= 1, & \Tr\left( A_{\mbf{p}} A_{\mbf{q}}\right) &
=d \delta^{(d)}_{\mbf{p},\mbf{q}},
}
implying
\ea{
\Tr\left(\Pi^{\pm}_{\mbf{p}} \Pi^{\pm}_{\mbf{q}}\right) &= \frac{1}{4} \left(d \pm 2 + d\delta^{(d)}_{\mbf{p},\mbf{q}}\right).
}
So the frame is symmetric.
\end{proof}
\begin{theorem}
\label{cor:condes}
Let $\Pi$ be a WH SIC fiducial projector in  odd dimension $d$ and let $e^{i\theta_{\mbf{p}}}$ be as  in Eq.~\eqref{eq:phsedef}.  Let $F$ be any matrix in $\GL(2,\mbb{Z}/d\mbb{Z})$ such that $\det F = 2^{-1}$.  Then 
\ea{
\Pi^{\pm} &= \frac{d\pm 1}{2d}\boldsymbol{1}_d \pm \frac{1}{2d}\sum_{\mbf{p}\neq 0} e^{2i\theta_{F\mbf{p}}} D_{\mbf{p}}
}
are the fiducial projectors for  WH covariant STFFs of rank $(d\pm 1)/2$.
\end{theorem}
\begin{remark}
Here $\GL(2,\mbb{Z}/d\mbb{Z})$ is the group of invertible $2\times 2$ matrices whose elements are integers mod $d$.  The symbol $2^{-1}$ denotes the multiplicative inverse of $2$ mod $d$.  
\end{remark}
\begin{proof}
We can write
\ea{
\Pi^{\pm} &= \frac{1}{2} \left( \boldsymbol{1}_d \pm A\right).
}
where
\ea{
A &= \frac{1}{d} \sum_{\mbf{p}} e^{2i\theta_{F\mbf{p}}} D_{\mbf{p}}.
}
Now let $H$ be the complex Hadamard matrix defined in Corollary~\ref{cor:had}.  Then
\ea{
H_{\mbf{p},\mbf{q}} &=\frac{1}{d} \tau^{-2\la \mbf{p},\mbf{q}\ra} e^{2i\theta_{\mbf{p}-\mbf{q}}}.
}
We  use this,  the fact that $\la K \mbf{p}, K\mbf{q}\ra = (\det K )\la \mbf{p}, \mbf{q}\ra$ for all $\mbf{p}$, $\mbf{q}\in \mbb{Z}^2$ and $K\in \GL(2,\mbb{Z}/d\mbb{Z})$, and the fact  that $H^2 = \boldsymbol{1}_d$, to deduce
\ea{
\sum_{\mbf{u}} \tau^{\la \mbf{u},\mbf{p}\ra} e^{2i(\theta_{F\mbf{u}}+\theta_{F(\mbf{p}-\mbf{u})})} 
&=\sum_{\mbf{u}} \tau^{2\la \mbf{u},F\mbf{p}\ra} e^{2i(\theta_{\mbf{u}}+\theta_{F\mbf{p}-\mbf{u}})} 
\nn
&= d^2\sum_{\mbf{u}} H_{F\mbf{p},\mbf{u}}H_{\mbf{u},\boldsymbol{0}}
\nn
&= d^2 \delta^{(d)}_{\mbf{p},\boldsymbol{0}}.
\label{eq:thm7lsteq}
}
The claim now follows by an application of  Lemma~\ref{lm:stffcond}.
\end{proof}
The fact that the sets constructed  in this theorem are STFFs with the maximal number of elements means they  are conical designs~\cite{Graydon}.  When they are rescaled by the factor $2/(d(d\pm 1))$ so as to make them  POVMs they become SIMs~\cite{Graydon}.  

There are some connections between the STFFs constructed here and the Wigner POVM constructed in ref.~\cite{Marcus2}.  Continue to assume $d$ is odd and let $P$ be the parity operator which acts on the basis featuring in Eq.~\eqref{eq:dopprops} by
\ea{
P|j\ra &= |-j\ra
}
Then
\ea{
P &= \frac{1}{d} \sum_{\mbf{p}} D_{\mbf{p}}.
}
In view of  Lemma~\ref{lm:stffcond} this means the operators
\ea{
\Pi^{\pm}_W &= \frac{1}{2}(\boldsymbol{1}_d \pm P),
}
like the operators $\Pi^{\pm}$ constructed in Theorem~\ref{cor:condes}, are the fiducial projectors for a pair of WH covariant rank $(d\pm 1)/2$ STFFs.  Rescaling them by $2/(d(d\pm 1))$ they define the POVMs,
\ea{
E^{\pm}_{\mbf{p}} &= \frac{1}{d(d\pm1)} \left( \boldsymbol{1}_d  \pm P_{\mbf{p}}\right)
}
where 
\ea{
P^{\vphantom{\dagger}}_{\mbf{p}} &= D^{\vphantom{\dagger}}_{\mbf{p}} P D^{\dagger}_{\mbf{p}}.
}
The associated probability distributions in the state $\rho$ are obtained from the Wigner function $W_{\mbf{p}}$ by shifting and re-scaling:
\ea{
\Tr(\rho E^{\pm}_{\mbf{p}}) &= \frac{1}{d\pm1} \left( \frac{1}{d} \pm W_{\mbf{p}}\right).
}
We accordingly refer to $\Pi^{\pm}_{\mbf{p}}$ as the Wigner STFFs.

We see from this that Eq.~\eqref{eq:Acond} has at least two  solutions: namely, the trivial solution $e^{i\phi_{\mbf{p}}}=1$, giving rise to the Wigner STFFs, and, as it is shown in the proof of Theorem~\ref{cor:condes}, the solution $e^{i\phi_{\mbf{p}}}=e^{2 i \theta_{F\mbf{p}}}$, giving rise to the SIC-related STFFs there described.  It will be shown in Eq.~\eqref{eq:NaimarkCompEqns})  below that these equations have non-SIC solutions .  Nevertheless, they are  potentially  interesting.  

Eq.~\eqref{eq:Acond} for $\mbf{p}=\boldsymbol{0}$ is automatic for any choice of the $e^{i\phi_{\mbf{p}}}$ satisfying $e^{i\phi_{-\mbf{p}}}=e^{-i\phi_{\mbf{p}}}$.  Written in terms of the SIC overlap phases $e^{i\theta_{\mbf{p}}}$ the other conditions read (by the first line of Eq.~\eqref{eq:thm7lsteq})
\ea{
\sum_{\mbf{u}}\tau^{2\la \mbf{u},\mbf{p}\ra} e^{2i(\theta_{\mbf{u}} + \theta_{\mbf{p}-\mbf{u}})} &= 0, & \forall \mbf{p} &\neq \boldsymbol{0} \mod d.
\label{eq:SICaltB}
}
It is natural to ask how many distinct solutions these equations have. At one extreme it might turn out that SICs are situated on a continuous manifold of solutions. At the other extreme it might be that the trivial solution and the SIC-related solutions are the only solutions.   In the former case a full description of the manifold could be a source of major insight.  One standard strategy for attacking a difficult problem is to relax the conditions, so as to embed the object of interest in a larger class of objects. Eqs.~\eqref{eq:SICaltB} have a very simple structure.  In spite of much effort no one has succeeded in solving the standard defining equations of a SIC.  Perhaps these equations will prove to be more tractable.

\section{Naimark Complement of a SIC} 
\label{sec:naimark}
In this section we consider the Naimark complement of a WH SIC (as opposed to one of the higher dimensional ETFs constructed from it in Section~\ref{sec:associated1}).  The construction differs from those presented in other sections in that it does not involve the squared overlap phases.  Its interest comes from the fact that it leads to another relaxation of the WH SIC existence problem.

It is easily seen that the Naimark complement of a WH SIC is covariant with respect to the direct sum of $d-1$ copies of the WH group.  Indeed, let $|\psi\ra$ be a WH SIC fiducial vector.  Let $|\phi_1\ra, \dots , |\phi_{d}\ra$ be an orthonormal basis for $d$ dimensional Hilbert space such that $|\phi_1\ra = |\psi\ra$ and define
\ea{
|\tilde{\psi}\ra &=\frac{1}{\sqrt{d-1}} \bmt |\phi_2\ra \\ |\phi_3\ra \\ \vdots \\ |\phi_{d} \ra \emt , & \tilde{D}_{\mbf{p}} &= \bmt D_{\mbf{p}} & 0  & \dots &0 \\0 & D_{\mbf{p}} & \dots  & 0  \\ \vdots &\vdots & \ddots & \vdots \\ 0&0&\dots& D_{\mbf{p}}\emt.
\label{eq:naimSICFid}
}
Using the fact that $\Tr(D_{\mbf{p}}) = d\delta_{\mbf{p},\boldsymbol{0}}$ it is  straightforward to verify that $\{\tilde{D}_{\mbf{p}}|\tilde{\psi}\ra \colon p_1, p_2 = 0 , \dots, d-1\}$ is a Naimark complement.  Let us note that this construction is  similar to one described in ref.~\cite{BodKing}.

It is a so-far unexplained fact~\cite{Zauner,Marcus,Scott,Andrew,FHS} that  every known WH SIC fiducial is an eigenvector of an order $3$ unitary $U$ which has the property
\ea{
UD_{\mbf{p}}U^{\dagger} = \tau^{2\la \boldsymbol{\xi}, \mbf{p}\ra} D_{F\mbf{p}}
}
for some ``vector'' $\boldsymbol{\xi}=\left(\begin{smallmatrix} \xi_1 \\ \xi_2 \end{smallmatrix}\right)$ and matrix $F=\left(\begin{smallmatrix} \alpha & \beta \\ \gamma & \delta\end{smallmatrix}\right)$ such that $\Tr(F) =-1$ mod $d$.  These two properties mean that $U$ permutes and rephases the elements of the SIC.
  It is easily seen that the Naimark complement has a similar order $3$ symmetry.   Indeed, we may choose the vectors $|\phi_j\ra$ so that they are an eigenbasis for $U$.  Let $\eta_j$ be the corresponding eigenvalues, and define
\ea{
\tilde{U} &= \bmt \eta^{*}_2 U & 0  & \dots & 0\\ 0 &  \eta^{*}_3 U & \dots & 0  \\ \vdots &\vdots & \ddots & \vdots\\ 0&0&\dots &  \eta^{*}_d U\emt.
}
 Then $\tilde{U} |\tilde{\psi}\ra = |\tilde{\psi}\ra$ and
\ea{
\tilde{U}\tilde{D}_{\mbf{p}}\tilde{U}^{\dagger} = \tau^{2\la \boldsymbol{\xi}, \mbf{p}\ra} \tilde{D}_{F\mbf{p}}.
}

If we allow the vector $|\tilde{\psi}\ra$ to range over the whole of $\mbb{C}^{d(d-1)}$ then the system of equations
\ea{
\bigl| \la \tilde{\psi} |\tilde{D}_{\mbf{p}} |\tilde{\psi}\ra \bigr|^2 &= \begin{cases} \frac{1}{(d-1)(d^2-1)} \qquad & \mbf{p}\ne \boldsymbol{0} \mod d \\ 1 \qquad &\mbf{p} = \boldsymbol{0} \mod d \end{cases}
\label{eq:NaimarkCompEqns}
}
gives us another relaxation of the WH SIC existence problem.  The case when $|\tilde{\psi}\ra$ is constrained to be the direct sum of $d-1$ orthonormal vectors (so that every solution is the fiducial vector of the complement of a WH SIC) is also interesting since it give us a system of equations which, while defining the same object, has a different algebraic structure from the usual defining equations for a SIC.  For instance, although the number of equations is $d^2$, as in the usual formulation, the number of real variables is  larger.  Indeed, in the usual formulation we need $2d-1$ variables to fix a single unit vector.  Since the phase cancels out of the equations the effective number of free parameters is $2d-2$.  By contrast in this formulation we need $2d-1$ real variables 
to fix  $|\phi_2\ra$ in Eq.~\eqref{eq:naimSICFid}, $2d-3$ variables to fix $|\phi_3\ra$ (two fewer because of the constraint that $|\phi_3\ra$ is orthogonal to $|\phi_2\ra$, \dots .  Taking account of the fact that the overall phase again cancels out of the problem this gives us $d^2-2$ variables in total.  We thus have
\ea{
 \lim_{d\to \infty}\left(\frac{\text{number of equations}}{\text{number of real variables}}\right) &=
 \begin{cases}
 \infty \qquad &\text{in the usual formulation}
 \\
 1 \qquad &\text{in this formulation}
 \end{cases}
}
(For the seeming 
overdetermination in the usual  formulation see also ref. \cite{FHS}.)

\section{Continuous Families}
\label{sec:contfam}
%
We now return to the constructions in Sections~\ref{sec:associated1} and~~\ref{sec:assoc2}, involving the squares of the overlap phases.
In Section~\ref{sec:associated1} we showed that a SIC in dimension $d$ is  associated to
\begin{enumerate}
\item a rank $d(d+1)/2$ projector $Q$ in dimension $d^2$,
\item a complex Hermitian Hadamard matrix $H$ in dimension $d^2$,
\item two ETFs $E$ and $\tilde{E}$ in dimensions $d(d+1)/2$ and $d(d-1)/2$.
\end{enumerate}
In Section~\ref{sec:assoc2} we showed that a WH SIC in odd dimension $d$ is also associated to
\begin{enumerate}
\setcounter{enumi}{3}
\item two WH covariant STFFs  $\Pi^{\pm}$ of rank $(d\pm 1)/2$.
\end{enumerate}
For $d\ge 4$ every known SIC is isolated. This is not always (possibly never) true of the structures $Q$, $H$, $E$, $\tilde{E}$, $\Pi^{\pm}$ associated to them.  We will show, by explicit construction, that the structures $Q$, $H$, $E$, and $\tilde{E}$ associated to   SICs on orbits\footnote{We are here using the classification scheme of ref.~\cite{Scott}.  Thus $4a$ and $6a$ are the unique extended Clifford orbits of SICs in dimensions 4 and 6 while $8b$ is one of the two such orbits in dimension 8.} $4a$, $6a$ and $8b$ in dimensions 4, 6 and 8 are points on continuous, one-parameter curves $Q(t)$, $H(t)$, $E(t)$, and $\tilde{E}(t)$ of projectors, complex Hadamard matrices, and ETFs respectively.  In  Appendix~\ref{Ap1} we present  calculations of the restricted defect  which encourage the conjecture, that for $d\ge 3$ these structures (along with the STFFs for WH SICs in odd dimension) are in fact always embedded in  manifolds of dimension 1 or higher.

In this section we are exclusively concerned with structures associated to a WH SIC.  For such a SIC, with normalized fiducial vector $|\psi\ra$, we can write
\ea{
Q_{\mbf{p},\mbf{q}} &= \frac{1}{2d} \tau^{-2\la \mbf{p},\mbf{q}\ra} M_{\mbf{p}-\mbf{q}} + \frac{1}{2}\delta_{\mbf{p},\mbf{q}},
\label{eq:QtermsM}
\\
H_{\mbf{p},\mbf{q}} &= \frac{1}{d} \tau^{-2\la \mbf{p},\mbf{q}\ra} M_{\mbf{p}-\mbf{q}},
\label{eq:HtermsM}
\\
E_{\mbf{p}, \mbf{q}} &= \frac{1}{d+1} \tau^{-2\la \mbf{p},\mbf{q}\ra} M_{\mbf{p}-\mbf{q}} + \frac{d}{d+1}\delta_{\mbf{p},\mbf{q}},
\label{eq:EtermsM}
\\
\tilde{E}_{\mbf{p}, \mbf{q}} &=- \frac{1}{d-1} \tau^{-2\la \mbf{p},\mbf{q}\ra} M_{\mbf{p}-\mbf{q}} + \frac{d}{d-1}\delta_{\mbf{p},\mbf{q}},
\label{eq:EttermsM}
}
where
\ea{
M_{\mbf{p}} &= e^{2i\theta_{\mbf{p}}}=
\begin{cases}
1 \qquad & \mbf{p} = \boldsymbol{0} \quad \text{(mod $d$)},
\\
(d+1)\left( \la \psi |D^{\dagger}_{\mbf{p}} |\psi\ra\right)^2 \qquad & \text{otherwise}.
\end{cases}
}
If $d$ is odd we also have the STFFs
\ea{
\Pi^{\pm} &= \frac{1}{2}\boldsymbol{1}_d \pm \frac{1}{2d} \sum_{\mbf{p}} M_{F\mbf{p}} D_{\mbf{p}}.
\label{eq:PiTermsM}
}
where $F$ is as defined in the statement of Theorem~\ref{cor:condes}.
We define a $d\times d$ matrix by
\ea{
M &= \bmt M_{0,0} & M_{0,1} & \dots & M_{0,d-1} \\ M_{1,0} &M_{1,1} & \dots & M_{1,d-1} \\ \vdots & \vdots & & \vdots \\ M_{d-1,0} & M_{d-1,1} & \dots & M_{d-1,d-1} \emt
}
where we use the shorthand $M_{p_1,p_2} = M_{(p_1,p_2)^{\rm{T}}}$.  We refer to $M$ as the squared-phase matrix.
The fact that $M$ is defined in terms of the squared matrix elements means we can work mod $d$ irrespective of whether $d$ is even or odd.  In particular
\ea{
M^{*}_{p_1,p_2} &= M\vpu{*}_{d-p_1,d-p_2}.
}
for all $\mbf{p}$, $d$.  

%

We now  lift the restriction to matrices $M$ defined in terms of a WH SIC.
\begin{lemma}
Let $M_{\mbf{p}} $ be any set of numbers such that
\ea{
M_{\boldsymbol{0}} = 1,
\label{eq:M0}
}
and
\ea{
|M_{\mbf{p}} | =1, \qquad M_{-\mbf{p}} =M^{*}_{\mbf{p}}
\label{eq:Mabs}
}
for all $\mbf{p}$.  
Let $Q$, $H$, $E$, $\tilde{E}$ be the $d^2\times d^2$ matrices defined in terms of $M$ by Eqs.~\eqref{eq:QtermsM}--\eqref{eq:EttermsM}.  If $d$ is odd also let $\Pi^{\pm}$ be the matrices defined in terms of $M$ by Eq.~\eqref{eq:PiTermsM}.  Then the following statements are equivalent
\begin{enumerate}
\item \label{it:Mprop} $M$ satisfies
\ea{
\sum_{\mbf{r}} \tau^{-2\la \mbf{p}, \mbf{r}\ra } M_{\mbf{p}-\mbf{r}} M_{\mbf{r}} &= d^2 \delta_{\mbf{p},\boldsymbol{0}}
\label{eq:MPropeq}
}
for all $\mbf{p}$,
\item \label{it:Qprop}  $Q$ is a rank $d(d+1)/2$ projector,
\item \label{it:Hprop}  $H$ is a complex Hermitian Hadamard matrix with $\Tr(H) = d$.
\item \label{it:Eprop}  $E$, $\tilde{E}$ are the Gram matrices of  ETFs in dimension $d(d+1)/2$, $d(d-1)/2$ respectively.
\end{enumerate}
If $d$ is odd these statements are also equivalent to
\begin{enumerate}[resume]
\item \label{it:Piprop}  $\Pi^{\pm}$ are fiducial projectors for a pair of rank  $(d\pm1)/2$ WH covariant STFFs
\end{enumerate}
\end{lemma}
\begin{proof}
Easy generalization of  arguments previously given.
\end{proof}

The lemma means that, in order to prove the existence of one-parameter families of projectors $Q(t)$ and complex Hadamard matrices $H(t)$ in dimension $d^2$, ETFs $E(t)$, $\tilde{E}(t)$ in dimensions $d(d\pm 1)/2$, and, when $d$ is odd, STFFs $\Pi^{\pm}(t)$ in dimension $d$, it is enough to prove the existence of a one-parameter family of $d\times d$  matrices  $M(t)$ satisfying Eqs.~\eqref{eq:M0}, \eqref{eq:Mabs}, and \eqref{eq:MPropeq}.
We now proceed to demonstrate the existence of such families for $d=4, 6, 8$. Our approach is trial-and-error:  i.e. to make an ansatz, and then to show it works.  It may be possible to prove (or disprove) existence in this way for a few more small values of $d$.  But to make progress for larger values a more systematic approach would be needed. 

\subsection{One parameter family for $d=3$}  The purpose of this section is to calculate hitherto unknown one parameter families in dimensions $4$, $6$, $8$, where every known SIC is isolated.  But for the sake of completeness let us note that in the special case of dimension $3$ a one-parameter family of SICs is already known~\cite{Zauner}, leading to a continuous family of squared-phase matrices.  The associated affine family of nine-by-nine complex Hadamard matrices intersects the non-affine family described in ref.~\cite{Karlsson}.
\subsection{One-parameter family for $d=4$}  Let $\Pi$ be the exact fiducial on orbit $4a$ described in ref.~\cite{AYAZ} (so $\Pi$ is an exact version of the Scott-Grassl~\cite{Scott} numerical fiducial on orbit $4a$, but differs from the Scott-Grassl exact fiducial).  Then the corresponding phase-matrix is
\ea{
M^{(0)} &= \bmt 1 & u^{-1} & 1 & u  \\ u^{-1} & u & u & u \\ 1 & u^{-1} & 1 & u \\ u & u^{-1} &u^{-1} & u^{-1}  \emt
}
where
\ea{
u &=-\sqrt{\frac{3-\sqrt{5}}{2}} + i \sqrt{\frac{\sqrt{5}-1}{2}}.
\label{eq:d4unit}
}
It turns  out  that this matrix continues to satisfy Eqs.~\eqref{eq:M0}, \eqref{eq:Mabs}, and \eqref{eq:MPropeq} even if the phase $u$ is left arbitrary.  So we can choose for our one paramter family
\ea{
M(t) &= \bmt 1 & e^{-it} & 1 & e^{it}  \\ e^{-it} & e^{it} & e^{it} & e^{it} \\ 1 & e^{-it}& 1 & e^{it} \\ e^{it} & e^{-it}&e^{-it}& e^{-it}  \emt, \qquad 0 \le t < 2 \pi
}

\subsection{One-parameter family for $d=8$}Let $\Pi$ be the exact fiducial on orbit $8b$ described in ref.~\cite{AYAZ} (as explained in ref.~\cite{AYAZ} this differs from the  Scott-Grassl~\cite{Scott}  numerical fiducial on orbit $8b$ as well as the exact one).  Then the corresponding squared-phase matrix is
\ea{
M^{(0)} &= \bmt 
1& u_1 u_2^{-1}& u_1^{2} & u_1^{-1} u_2 & 1 & u_1 u_2^{-1}& u_1^{-2}  & u_1^{-1} u_2
\\
u_1 u_2^{-1} & u_1^{-1} u_2 & u_2 & u_1 u_2^{-1} & u_2 & u_2^{-1} & u_1 u_2^{-1} & u_2
\\
u_1^{2} & u_2^{-1} & u_1^{-2} & u_1u_2^{-1} & u_1^{2} & u_2^{-1} & u_1^{2} & u_1 u_2^{-1}
\\
u_1^{-1} u_2 & u_1^{-1} u_2 & u_1^{-1} u_2 & u_1 u_2^{-1} & u_2^{-1} & u_2^{-1} & u_2^{-1} & u_2 
\\
1 & u_2 & u_1^{-2} & u_2^{-1} & 1 & u_2 & u_1^{2} & u_2^{-1}
\\
u_1u_2^{-1} & u_2^{-1} & u_2 & u_2 & u_2 & u_1^{-1} u_2 &  u_1 u_2^{-1}  &u_1 u_2^{-1}
\\
u_1^{-2} & u_1^{-1} u_2 & u_1^{-2} & u_2 & u_1^{-2} & u_1^{-1} u_2 & u_1^{2} & u_2
\\
u_1^{-1} u_2 & u_2^{-1} & u_1^{-1} u_2 & u_2 & u_2^{-1} & u_1^{-1} u_2 & u_2^{-1} & u_1 u_2^{-1}
\emt,
\label{eq:8bPM}
}
where
\ea{
u_1&= -\sqrt{\frac{3-\sqrt{5}}{2}} + i \sqrt{\frac{\sqrt{5}-1}{2}},
\\
u_2 &= -\frac{1}{2}\sqrt{17+8\sqrt{2}-7\sqrt{5}-4\sqrt{10}}+\frac{i}{2}\sqrt{-13-8\sqrt{2}+7\sqrt{5}+4\sqrt{10}}.
}
It will be observed that $u_1$ is the same as the phase $u$ in Eq.~\eqref{eq:d4unit}.  This is connected with the phenomenon of dimension towers~\cite{AFMY1,AFMY2,ABDF,Markus2}:  Specifically with the fact that $8b$ is the fiducial aligned with $4a$ in the sequence $4, 8, 48, \dots$. 

In view of the result for $d=4$ the obvious ansatz  is to retain the form of the matrix in Eq.~\eqref{eq:8bPM}, while allowing $u_1$ and $u_2$ to vary independently.  However, this does not work.  Instead we have to make a further relaxation, and replace $u_1^2$ with a third independent phase factor $u_3$ to obtain the ansatz
\ea{
M(u_1,u_2,u_3) &= \bmt 
1& u_1 u_2^{-1}& u_3 & u_1^{-1} u_2 & 1 & u_1 u_2^{-1}& u_3^{-1}  & u_1^{-1} u_2
\\
u_1 u_2^{-1} & u_1^{-1} u_2 & u_2 & u_1 u_2^{-1} & u_2 & u_2^{-1} & u_1 u_2^{-1} & u_2
\\
u_3 & u_2^{-1} & u_3^{-1} & u_1u_2^{-1} & u_3 & u_2^{-1} & u_3 & u_1 u_2^{-1}
\\
u_1^{-1} u_2 & u_1^{-1} u_2 & u_1^{-1} u_2 & u_1 u_2^{-1} & u_2^{-1} & u_2^{-1} & u_2^{-1} & u_2 
\\
1 & u_2 & u_3^{-1} & u_2^{-1} & 1 & u_2 & u_3& u_2^{-1}
\\
u_1u_2^{-1} & u_2^{-1} & u_2 & u_2 & u_2 & u_1^{-1} u_2 &  u_1 u_2^{-1}  &u_1 u_2^{-1}
\\
u_3^{-1} & u_1^{-1} u_2 & u_3^{-1} & u_2 & u_3^{-1} & u_1^{-1} u_2 & u_3 & u_2
\\
u_1^{-1} u_2 & u_2^{-1} & u_1^{-1} u_2 & u_2 & u_2^{-1} & u_1^{-1} u_2 & u_2^{-1} & u_1 u_2^{-1}
\emt.
\label{eq:8bPMR}
}
 The need for this further relaxation may  be related to the fact that $\{u_2,u^{-1}_2,u_1u_2^{-1},u_1^{-1}u_2\}$ and $\{u_1^2, u_1^{-2}\}$ are different orbits under the action of the Galois group $\Gal\big(\mbb{E}/\mbb{Q}(a)\big)$ (where $\mbb{E}$ and $a$ are as defined in ref.~\cite{AYAZ}).  The matrix satisfies Eqs.~\eqref{eq:M0} and \eqref{eq:Mabs} for any choice of $u_1$, $u_2$, $u_3$.  It is convenient to define phase angles $\psi$, $\phi$, $\theta$ by
\ea{
u_1 = e^{2i\psi}, \qquad u_2 = e^{i(\phi+\psi)}, \qquad u_3 = e^{i\theta}.
} 
Inserting these expressions into Eq.~\eqref{eq:MPropeq} 
and solving for $\phi$, $\theta$ in terms of $\psi$ one finds, after lengthy but straightforward manipulations, that the equations are satisfied if and only if $M$ is on one of the three curves
\ea{
M(e^{2i\psi},e^{i\psi}, 1), \qquad M(e^{2i\psi},e^{i\psi}f(\psi), g(\psi)), \qquad M(e^{2i\psi},e^{i\psi}f(\psi), (g(\psi))^{*})
}
where
\ea{
f(\psi) &= \frac{\sqrt{1+2\sin^2\psi + 4 \sin^4 \psi} + \sqrt{2}i\sin \psi}{1+2\sin^2\psi}
\\
g(\psi) &= \frac{(1-2\sin^2\psi) + 2\sqrt{2} i \sin \psi}{1+2\sin^2\psi}
}
and $\psi$ is arbitrary (to obtain these expressions it is enough to solve Eq.~\eqref{eq:MPropeq}  for $\mbf{p}=(0,1)^{\rm{T}}$, $(0,2)^{\rm{T}}$, $(0,4)^{\rm{T}}$; the remaining equations are then satisfied automatically).    The  curves intersect at the points $\psi =0$, $\pi$ and nowhere else. 
 
 \subsection{One-parameter family for $d=6$}Let $\Pi$ be the strongly centred exact fiducial on orbit $6a$ described in ref.~\cite{ACFW} (as explained in ref.~\cite{ACFW} this differs from  the exact fiducial described in ref.~\cite{AYAZ} as well as  the numerical and exact fiducials described in ref~\cite{Scott}).  The corresponding squared-phase matrix is
 \ea{
 M^{(0)} &= 
 \bmt
 1& u_1& u_3& 1& u_3^{-1}& u_1^{-1}
 \\
 u_1& u_1^{-1}& u_3^2& u_2& u_1  u_2  u_3^{-2}&  u_3^2
 \\
 u_3& u_3^{-2}& u_3^{-1} & u_1 u_2  u_3^{-2}& u_3^3& u_2
 \\
 1& u_1^{-1} u_2^{-1}  u_3^2& u_2^{-1}& 1& u_2& u_1 u_2  u_3^{-2}
 \\
 u_3^{-1}& u_2^{-1}& u_3^{-3} & u_1^{-1} u_2^{-1} u_3^{2} & u_3& u_3^2
 \\
 u_1^{-1}& u_3^{-2} & u_1^{-1}  u_2^{-1}  u_3^2& u_2^{-1}& u_3^{-2}& u_1
 \emt
 \label{eq:d6pmt}
 }
 where
 \ea{
 u_1 &= \sqrt{c_1} +i \sqrt{1-c_1},
 \\
 u_2 &=- \sqrt{c_2} + i \sqrt{1-c_2},
 \\
 u_3 &=  -\sqrt{c_3}+i\sqrt{1-c_3},
 }
 \ea{
 c_1 &= \frac{567-85\sqrt{21}}{168} + \frac{-63+17\sqrt{21}}{42} b + \frac{4(21-5\sqrt{21})}{21} b^2,
 \\
 c_2&= \frac{1743-349\sqrt{21}}{168} +\frac{-63+11\sqrt{21}}{42} b + \frac{-63+13\sqrt{21}}{21}b^2,
 \\
 c_3 &= \frac{11-\sqrt{21}}{8},
 \\
b&=  \rel\Big( \left(1+i\sqrt{7}\right)^{\frac{1}{3}}\Big).
\label{eq:6aPM}
 }
 Similarly to the situation with fiducial $8b$, the obvious ansatz, which consists in retaining the form of the matrix in Eq.~\eqref{eq:d6pmt} while allowing $u_1$, $u_2$, $u_3$ to vary independently, does not work.  Instead we have to make a further relaxation, by making the replacements
\eat{3}{
u_1^{\pm 1} &\to v_1^{\pm 1},
&\qquad
u_2^{\pm 1} &\to v_2^{\pm 1},
&\qquad
u_1^{\pm 1}u_2^{\pm 1}u_3^{\mp2} &\to v_3^{\mp 1}
\\
u_3^{\pm 1} &\to v_4^{\pm 1}
&\qquad
u_3^{\pm 2} &\to v_5^{\pm 1}
&\qquad
u_3^{\pm 3} & \to v_6^{\pm 1}
}
where $v_1$, \dots, $v_6$ are independent phases, giving us the ansatz
\ea{
M(v_1,v_2,v_3,v_4,v_5,v_6)&=
 \bmt
 1& v_1& v_4& 1& v_4^{-1}& v_1^{-1}
 \\
 v_1& v_1^{-1}& v_5 & v_2& v_3^{-1}&  v_5
 \\
 v_4& v_5^{-1}& v_4^{-1} & v_3^{-1}& v_6 & v_2
 \\
 1& v_3 & v_2^{-1}& 1& v_2&v_3^{-1}
 \\
 v_4^{-1}& v_2^{-1}& v_6^{-1} & v_3 & v_4& v_5
 \\
 v_1^{-1}& v_5^{-1} & v_3 & v_2^{-1}& v_5^{-1}& v_1
 \emt
}
 The need for this  relaxation may  be related to the fact that $\{u_1^{\pm 1}, u_2^{\pm 1} , u_1^{\pm 1}u_2^{\pm 1}u_3^{\mp 2}\}$,  $\{u_3^{\pm1}\}$, $\{u_3^{\pm 2}\}$, $\{ u_3^{\pm 3}\}$ are different orbits under the action of the Galois group $\Gal\big((\mbb{E}/\mbb{Q}(a)\big)$ (where $\mbb{E}$ and $a$ are as defined in ref.~\cite{AYAZ}) (c.f. our analogous remark concerning the family based on fiducial $8b $). Let 
 \ea{
 S(\mbf{p}) &=\sum_{\mbf{r}} \tau^{-2\la \mbf{p}, \mbf{r}\ra } M_{\mbf{p}-\mbf{r}} M_{\mbf{r}}
 }
We require that $S(\mbf{p}) = 36 \delta_{\mbf{p},\boldsymbol{0}}$ for all $\mbf{p}$.  Aside from the condition $S(\boldsymbol{0})=36$, which is trivial, each of these equations  coincides, either with one of   the seven  equations $S(\mbf{p}_j) = 0$, $j=1, 2, \dots, 7$, or with its conjugate,  where
\ea{
\mbf{p}_1 &= \bmt 2 \\ 4 \emt,  \qquad \mbf{p}_2 =\bmt 1 \\ 2 \emt,  \qquad \mbf{p}_3 = \bmt 0 \\ 2\emt ,\qquad \mbf{p}_4 = \bmt 0 \\ 3\emt, 
}
\ea{
\mbf{p}_5 &=\bmt 0 \\ 5\emt , \qquad  \mbf{p}_6 = \bmt 5 \\ 3 \emt ,\qquad \mbf{p}_7 = \bmt 3 \\ 5 \emt.
}
This system of equations is  more complicated than the ones considered earlier, and the obvious approach is to construct a Gr\"{o}bner basis.  This is in fact the procedure we originally adopted.  However,  to get a  solution in an acceptable amount of time using available resources one needs  to do some  pre-processing and it turns out that once that has been done it is actually quicker and easier to complete the calculation by hand.  We  indicate the main  steps.

We begin by observing that  cyclic permutations of the variables $v_1$,  $v_2$, $v_3$ leave $S(\mbf{p}_1)$, \dots, $S(\mbf{p}_4)$ invariant and cyclically permute $S(\mbf{p}_5)$, $S(\mbf{p}_6)$, $S(\mbf{p}_7)$ (a symmetry stemming  from the Galois symmetries of the original SIC).    This suggests that we work in terms of the symmetric polynomials
\ea{
s_1 = v_1 + v_2 + v_3, \qquad s_2 = v_1 v_2 + v_2 v_3 + v_3 v_1, \qquad s_3 = v_1v_2 v_3.
}
The fact that the equations are not invariant under odd permutations of $v_1$,  $v_2$, $v_3$ suggests that we also define
\ea{
s_4 &= i (v_1-v_2)(v_2-v_3)(v_3-v_1)
}
(where the additional factor of $i$ is introduced for later convenience).  If we  define
\ea{
E_j &= S(\mbf{p}_j)
}
for $j=1,\dots 4$, and
\ea{
E_5 &= S(\mbf{p}_5) +  S(\mbf{p}_6) + S(\mbf{p}_7),
\\
E_6&= \big(S(\mbf{p}_5)-S(\mbf{p}_6)\big)\big(S(\mbf{p}_6)-S(\mbf{p}_7)\big) +\big(S(\mbf{p}_6)-S(\mbf{p}_7)\big)\big(S(\mbf{p}_7)-S(\mbf{p}_5)\big)
\nn
&\hspace{7 cm}+\big(S(\mbf{p}_7)-S(\mbf{p}_5)\big)\big(S(\mbf{p}_5)-S(\mbf{p}_6)\big),
\\
E_7 &= \big(S(\mbf{p}_5)-S(\mbf{p}_6)\big)\big(S(\mbf{p}_6)-S(\mbf{p}_7)\big) \big(S(\mbf{p}_7)-S(\mbf{p}_1)\big),
}
then the equations $E_1=0$, \dots, $E_7=0$ are equivalent to the original set of  equations, and are completely expressible in terms of the seven variables $s_1, \dots, s_4, v_4, v_5, v_6$.

It is easily verified that the fact that the $v_j$ are phases implies
\ea{
|s_1| = |s_2|, \qquad |s_3| = 1, \qquad s_1 s_2 = |s_1 s_2| s_3, \qquad \frac{s_4}{s_3} \in \mbb{R}
}
meaning we can write
\ea{
s_1 = r e^{i\phi_1}, \qquad s_2 = r e^{i\phi_2}, \qquad s_3 = e^{i(\phi_1+\phi_2)}, \qquad s_4  = k e^{i(\phi_1+\phi_2)},
\label{eq:s1s2s3s4conds}
} 
for suitable $r, \phi_1, \phi_2, k\in \mbb{R}$ with $r$ non-negative.  The condition that the $s_j$ be expressible in this form is  sufficient as well as necessary for $v_1$, $v_2$, $v_3$ to be phases.  Indeed, suppose that   Eq.~\eqref{eq:s1s2s3s4conds} holds for some $r, \phi_1, \phi_2, k$  satisfying the stated conditions.  It is easily seen that  $v_j$ and $v^{*-1}_j$ are both roots of the polynomial $x^3-s_1 x^2+s_2 x -s_3$ for $j=1, 2, 3$,  which means that if  $v_1$, $v_2$, $v_3$ were not all phases  they would have to be a cyclic permutation of the three numbers 
\ea{
\mu e^{i\xi}, \qquad  \frac{e^{i\xi}}{\mu}, \qquad e^{i(\phi_1+\phi_2-2\xi)},
}
for some $\xi, \mu \in \mbb{R}$ such that $\mu > 0$ and $\mu \neq 1$.
But then
\ea{
k &=i \left(\mu-\frac{1}{\mu}\right)\left(\mu+\frac{1}{\mu} +2 \cos (3\xi -\phi_1 -\phi_2)\right)
}
would be pure imaginary.

Writing $v_5 = e^{i\theta_5}$, $v_6 = e^{i\theta_6}$ and taking the imaginary part of the equation $E_1+3 E_2=0$ gives
\ea{
\left( 3\cos \theta_5-\cos \theta_6 -2\right)\left(3 \sin \theta_5+\sin \theta_6\right) &= 0
\label{eq:firstFactors}
}
This potentially gives us two families of solutions, according to whether we set the first or second factor to zero.  We choose to impose the  condition $3\cos\theta_5 = \cos \theta_6 +2$, as this is the one satisfied by the original SIC\footnote{We neglect the other possibility, not because it is not potentially interesting, but simply for reasons of time.}.  The condition implies
\ea{
v_5 &= \frac{3+t^2 +i (-1)^{n_1}2^{\frac{3}{2}}  t \sqrt{t^2 -3}}{3(t^2-1)}
\\
v_6 &= \frac{5-t^2 +i(-1)^{n_1}2^{\frac{3}{2}}  \sqrt{t^2-3}}{t^2-1}
}
for some real parameter $t \in \left(-\infty,-\sqrt{3}\, \right] \cup \left[\sqrt{3},\infty\right)$ and integer $n_1 =0$ or $1$, and where we take the principal branch for the square root.
Taking the real part of $E_1=0$ and  imaginary part of $E_2=0$ gives, respectively,
\ea{
r&= \frac{t^2+15}{3(t^2-1)}, 
\\
k&= \frac{(-1)^{n_1}2^{\frac{5}{2}} (3-t) (t^2-3)^{\frac{3}{2}}}{9(t^2-1)^2},
}
where we take the principal branch for the fractional power.
The cases $t^2=3$ and $t^2>3$ need to be handled separately.  In the former case it is easily seen that the full solution is
\ea{
v_1 = v_2 =v_3 = v^{*}_4 = e^{i\phi}, \qquad v_5=v_6=1,
\label{eq:v14ts3}
}
for some arbitrary phase $e^{i\phi}$.  Turning to the second case, the equation $E_3 =0$ together with the assumption $t^2>3$ gives
\ea{
v_4 &= -\frac{(t^2+15)e^{-i\phi_1} }{9(t^2-1)}+\frac{(t^2+15)e^{i\phi_2}}{4(t^2-3)}
-\frac{(t^2+15)^2e^{2i\phi_1}}{36(t^2-3)(t^2-1)}.
\label{eq:v4expn}
}
Combining this equation with the real part of $E_2=0$ and taking account of the fact that $v_4$ is a phase gives, after lengthy but straightforward calculations, 
\ea{
e^{i\phi_1} &=\sigma^{n_4} \left(\frac{(e^{i\psi}p_1+p_2)\big( p_3 + i (-1)^{n_3}p_4 \big)}{p_5} \right)^{\frac{1}{3}}
\label{eq:f1}
\\
e^{i\phi_2} &= \sigma^{2n_4} e^{-i\psi} \left(\frac{(e^{i\psi}p_1+p_2)\big( p_3 + i(-1)^{n_3} p_4 \big)}{p_5} \right)^{\frac{2}{3}}
\label{eq:f2}
}
where 
\ea{
\sigma&=e^{\frac{2\pi i}{3}}, \qquad e^{i\psi} = p_6+i(-1)^{n_2} p_7,  \qquad n_2, n_3\in \{0,1\}, \qquad n_4\in \{0,1,2\}, 
\\
p_1 &= -9(t^2-1),
\\
p_2&= t^2+15,
\\
p_3 &=8(-207+108t-42t^2 -36t^3 +t^4),
\\
p_4 &= 24\sqrt{2}  (3-t)\left(-3+6t+t^2\right)\sqrt{t^2-3},
\\
p_5&=64(t^2-3)(171-108t-6t^2 +36t^3 +19t^4),
\\
p_6&=-\frac{1647-5103t+2403t^2+1053t^3-531t^4-621 t^4+65t^6+63t^7}{(t-1)(t^2+15)^3} 
\\
p_7&= \frac{2^{\frac{7}{2}} (t^2-3)^{\frac{3}{2}} \sqrt{(3-t)(1+t)(-837+1134t-135t^2-108t^3+45t^4+126t^5+31t^6}}{(t-1)(t^2+15)^3}
}
and where we take the principal branch for the fractional powers and  square roots.  Given the requirement that $t^2>3$ the necessary and sufficient condition for $e^{i\phi_1}$, $e^{i\phi_2}$ to be phases is that
\ea{
t \in J = [-t_0,-\sqrt{3})\cup(\sqrt{3},3]
}
where
\ea{
t_0 &= \frac{3}{31}\left(7-4\left(c+\frac{1}{c}\right) + 6\left(c^2+\frac{1}{c^2}\right)
\right)+\frac{2}{31}\sqrt{633+90\left(c+\frac{1}{c}\right) +144\left(c^2+\frac{1}{c^2}\right)}
}
with  $c= (2+\sqrt{3})^{\frac{1}{3}}$.  The roots of the cubic $x^3-s_1 x^2+s_2 x-s_3$ are
\ea{
\nu(j)&= \frac{1}{3}\left(s_1 +\sigma^j \Delta +\frac{B}{\sigma^j \Delta}\right), \qquad j=0,1,2
}
where
\ea{
\Delta &= \left(\frac{A + \sqrt{A^2 - 4B^3}}{2}\right)^{\frac{1}{3}},
\label{eq:DeltaDef}
\\
A&=  2s_1^3 -9s_1 s_2 +27 s_3,
\\
B&= s_1^2 -3 s_2 .
}
and where one is free to make any fixed choice of  branch  for the square and cube roots.  With the appropriate choice of branch for the square root we get the more explicit formula
\ea{
\Delta &= e^{i\phi_1} \left(  \frac{(15+t^2)^3}{27(t^2-1)^3} +e^{-i\psi} \left(\frac{-99-42t^2+13t^4}{(t^2-1)^2} + \frac{8\sqrt{6}(3-t)(t^2-3)^{\frac{3}{2}}}{3(t^2-1)^2} \right) \right)^{\frac{1}{3}}
}
where we now take the principal branch for the cube root.  For this choice of $\Delta$ one finds
\ea{
\frac{(v_1-v_2)(v_2-v_3)(v_3-v_1)}{(\nu(0)-\nu(1))(\nu(1)-\nu(2))(\nu(2)-\nu(0))} &= (-1)^{n_1}.
}
implying
\ea{
v_j &=\nu\big(n_5+(-1)^{n_1}( j-1)\big)
}
for $n_5 \in \{0,1,2\}$.  The fact that $r$, $\phi_1$, $\phi_2$, $k$ are real means, in view of the result proved earlier, that $v_1$, $v_2$, $v_3$ are phases.  

We derived this solution using the three conditions $E_1=E_2=E_3=0$.  The calculations  to show it satisfies all seven equations $E_j=0$, though tedious, present no essential difficulty.

Finally, let us note that it follows from Eq.~\eqref{eq:v14ts3} that the 144 curve segments corresponding to the 72 possible choices of  $n_1$, \dots, $n_5$ and 2 possible choices of $\sgn(t)$ are all connected at the points $t=\pm\sqrt{3}$.  

\subsection{Conclusion to this Section}
The results presented in this Section, combined with those in Appendix~\ref{Ap1}, suggest the conjecture that the squared-phase matrix for every SIC in dimension greater than 2 lies on a continuous curve.  It is possible one could confirm the conjecture in a few more cases using brute-force methods, such as we have employed here.  However, decisive progress would require new insights.

It is to be noted that what we have presented gives one-parameter families
of Gram matrices of ETFs. This is enough to guarantee existence, but for many
purposes it is preferable to know the vectors explicitly. In the case of $d = 4$ this
is quite easy to provide. We know~\cite{ABDF} that there is a subgroup of the
Weyl-Heisenberg group in dimension $d(d-2) = 8$ which, when acting on an $8b$ SIC
vector in this dimension, provides an ETF consisting of $d^2 = 16$ vectors sitting in a
$d(d-1)/2 = 6$ dimensional subspace of $\mathbb{C}^8$. Let us call this an ETF of type $(6, 16)$.
In this case it is not hard to find a one-parameter family of fiducial vectors in a
two dimensional subspace of $\mathbb{C}^8$, yielding a one-parameter family of $(6, 16)$ ETFs
as orbits under that subgroup of the Weyl-Heisenberg group. In fact this curve
is a geodesic in projective space, as is the one-parameter family of SIC vectors
in dimension $d = 3$. The resulting Gram matrices are equivalent to those arising
from the construction presented here. The situation appears to be considerably
more complicated for the non-affine families we obtained for $d = 8$ and $d = 6$, and
we have not made any progress in these cases.

Focussing on the family of $(6, 16)$ ETFs that we found in this way, we may
compare it with the Steiner ETF of the same type~\cite{Mixon}. The latter consists of sparse vectors, which is an advantage in some
applications. It can be shown to belong to a seven parameter family of (6, 16) ETFs~\cite{Turek}.
It is natural to ask if this family is connected
to our one-parameter family. The answer is that the two families do intersect, but
not at the SIC point---clearly not, since the restricted defect that we calculate in
the Appendix equals 1 there, and indeed this is the value taken by the restricted
defect at generic points on the curve. This observation answers a question raised
by Waldron~\cite{Waldron}.  Interestingly, Weyl-Heisenberg covariant Steiner ETFs were found recently~\cite{BodKing}.

 \section{Conclusion}
In this paper we have described a number of related structures associated to a given SIC:  the complex Hadamard matrix described in Corollary~\ref{cor:had}, the two ETFs described in Corollary~\ref{cor:ETF}, and, in the case of a WH SIC, the two STFFs described in Corollary~\ref{cor:condes}.   These structures were originally discovered as a result of the number-theoretic investigations described in ref.~\cite{ABDF}.  What is noteworthy here is that the squared overlap phases, whose importance is naturally suggested by the fact that a SIC is a projective $2$-design, also seem to be playing a vital role in the number-theoretic structures underlying SICs.  These connections between the number theory and geometry are intriguing, and they merit further investigation.  We also  described two sets of equations which are different from the standard defining equations of a SIC.  The fact that they have a different algebraic structure from the standard equations, and the fact that the first set of equations provides a relaxation of the SIC problem, means that they may be a potential source of new insight.  Finally, we showed that these structures belong to continuous families in three cases, and in Table~\ref{TableRD} of the appendix  we present  evidence that the same is true for  twenty three other cases for dimension $d\ge 4$.  It is already known that such families exist when $d=3$.  This suggests the conjecture that such families exist for every value of $d$ greater than 2.

%

\appendix
\section{Restricted Defect Calculations for $d=2$--$16$.}\label{Ap1}
Given a geometrical constellation of complex vectors, e.g. an equiangular tight frame (ETF) \cite{STDH07}, we might ask for the possibility to introduce non-trivial free parameters in order to generate a continuous family of solutions. Here, non-trivial is used in the sense that none of the free parameters can be absorbed by a rigid rotation acting over the entire constellation, or by global phases acting on vectors. The problem to determine existence and construction for such families has demonstrated to be highly non-trivial \cite{FMT11,Mixon,FJMP18}. Recently, the concept of restricted defect has been derived, which allows us to establish an upper bound for the maximal number of parameters allowed by a family of POVM having a prescribed geometrical structure \cite{Bruzda}. In particular, in the same work the authors have proven that known maximal sets of MUB (mutually unbiased bases) and SICs are isolated up to dimension 16. 

Let us briefly summarize how the method works. The Gram matrix $G$ of a tight frame composed by $N$ elements in dimension $d$ is proportional to a rank-$d$ projector. This implies that a suitable linear combination between $G$ and the identity operator produces a structured unitary matrix, i.e. a unitary matrix having prescribed absolute value of entries that only depend on the parameters $N$ and $d$. A continuous set of inequivalent tight frames is thus one-to-one connected with a continuous family of unistochastic matrices \cite{Turek}. 

To find the maximal family of tight frames arising from a given solution consists in finding the most general real antisymmetric matrix $R$ of size $N$ such that 
\begin{equation}\label{UN}
V_{j,k}(t)=U_{j,k}e^{itR_{j,k}},
\end{equation}
is a Hermitian unitary matrix, provided that $U=\mathbb{I}_N-\frac{2d}{N}G$ is associated to a given tight frame, e.g. SIC-POVM. The parameter $t$ in Eq.~\eqref{UN}) is introduced for convenience and it can be set to $t=1$ after applying our method. To calculate the exact number of free parameters we need to solve the following non-linear problem: \medskip

\begin{table}[h]
\begin{equation*}
\begin{array}{c|c||c|c}
\mbox{Fiducial state}&\mbox{Restricted defect}&\mbox{Fiducial state}&\mbox{Restricted defect}\\
\hline
2a  &0&11a&1\\
3a  &2 &11b&1\\
3b&4&11c&1\\
3c  &4&12a&1\\
4a &1&12b&40\\
5a&1&13a&1\\
6a&1&13b&1\\
7a&1&14a&1\\
7b&16&14b&1\\
8a &1&15a&226\\
8b&1&15b&226\\
8H&945&15c&226\\
9a&1&15d&424\\
9b&1&16a&1\\
10a&1&16b&1
\end{array}
\end{equation*}
\caption{Restricted defect for families of equiangular tight frames arising from known SICs in dimensions $d=2-16$.   Its value represents an upper bound for the maximal possible number of free parameters of a family containing the given particular solution. Note that there is a single isolated case ($d=2$). Also, for several cases the upper bound is 1. For unknown reasons, some ETF have a large (7b and 12b) or very large ($8H, 15a-15d$) defect. Aside from the label $8H$ (which refers to the Hoggar lines~\cite{Hoggar}) the SICs are labelled using the Scott-Grassl  scheme~\cite{Scott}.  The label $3a$ refers to an infinite family of unitarily inequivalent SICs which was sampled by varying the parameter $t$ in the parameterization of ref.~\cite{Marcus} in steps of $\pi/300$.}\label{TableRD}
\end{table}

\noindent \textbf{Problem }$\mathbf{\mathcal{P}_{NL}:}$ Find the most general matrix $V(t)$ of size $N$ having the form (\ref{UN}) such that
\begin{equation}\label{P1}
V(t)V(t)^{\dag}=V^2(t)=\mathbb{I}_N,
\end{equation}
This is a non-linear system of coupled equations depending on $t$ and on the $[N (N - 1)/2-z] - (N - 1)$ different variables $R_{ij}$, where $z$ is the number of zeros existing in the strictly upper triangular part of the matrix $U$. The resolution of Eq.~\eqref{P1} provides the most general family of structured tight frames. However, this is a highly complicated problem, in general. In order to simplify its resolution, we define the following linearized problem:\medskip

\noindent\textbf{Problem }$\mathbf{\mathcal{P}^{(2)}_L:}$ Find the most general matrix $V(t)$ of size $N$ such that
\begin{equation}\label{P2}
\lim_{t\rightarrow0}\frac{d}{dt}[V^2(t)]=0.
\end{equation}
Using (\ref{UN}), we can explicitly write Eq.(\ref{P2}) as
\begin{equation}\label{soe}
-2V_{k,k}(0)V_{k,j}(0)R_{j,k}+\sum_{l\neq j,k}V_{k,l}(0)V_{l,j}(0)(R_{k,l}-R_{j,l})=0,
\end{equation}
for $1\leq j<k\leq N$ and $1\leq l\leq N$, which is a linear system of equations in variables $R_{ij}$.
Note that $\mathcal{P}^{(2)}_L\subset\mathcal{P}^{(1)}_{NL}$, as Eq.(\ref{P2}) is a necessary condition to obtain Eq.~\eqref{P1}. For this reason, the dimension of the solution space of this linear system represents an upper bound for the maximal allowed number of free parameters in the tight frame, i.e. the restricted defect.

In Table \ref{TableRD}, we derive the restricted defect for the families of ETF arising from SICs in low dimensions, according to the results derived in Section~\ref{sec:associated1}. 

\section*{Acknowledgements}
Discussions with Gary McConnell played a crucial role at one stage of this work. We also thank Karol \.Zyczkowski for useful discussions, and two anonymous referees for some very helpful coments.. DG acknowledges Grant FONDECYT Iniciaci\'{o}n number 11180474, Chile.  This work was supported in part by the Australian Research Council through the Centre of Excellence in Engineered Quantum Systems CE170100009.

\end{document}